\documentclass[runningheads]{llncs}

\usepackage{latexsym}
\usepackage{amsmath}
\usepackage{amssymb}
\usepackage{proof}
\usepackage{mathtools}
\usepackage{tikz}
\usepackage{tikz-qtree}
\usepackage[all]{xy}

\newcommand{\fundef}[2]{#1 & = & \expr{#2}}
\newcommand{\expr}[1]{#1{}}
\newcommand{\var}[2]{\mathit{#1}#2}

\newcommand{\abs}[3]{\lambda #1{.#2{#3}}}
\newcommand{\app}[3]{#1{~#2{#3}}}
\newcommand{\args}[3]{#1{\ldots#2{#3}}}

\newcommand{\cas}[6]{\!\begin{array}[t]{llcl}\multicolumn{4}{l}{{\bf case} ~ #1{~{\bf of}}}\\& #2{} & \Rightarrow & #3{}\\~~~~ & #4{} & \Rightarrow & #5{#6}\end{array}}
\newcommand{\Cas}[6]{{\bf case} ~ #1{~ {\bf of} ~ #2{~ \Rightarrow #3{~; ~ #4{~ \Rightarrow #5{#6}}}}}}
\newcommand{\casedots}[6]{{\bf case}~#1{~{\bf of}~#2{\Rightarrow #3{\ldots#4{\Rightarrow #5{#6}}}}}}

\newcommand{\where}[3]{\begin{array}[t]{lcl}\multicolumn{3}{l}{#1{}}\\\multicolumn{3}{l}{{\bf where}}\\#2{#3}\end{array}}

\newcommand{\wheredots}[6]{#1{} ~ {\bf where} ~ #2 = #3 \ldots #4 = #5{#6}}

\newcommand{\Letexp}[4]{{\bf let} ~ #1{ =  #2{~ {\bf in} ~ #3{#4}}}}

\newcommand{\multiletexp}[6]{{\bf let} ~ #1{= #2{\ldots #3{= #4{~ {\bf in}~ #5{#6}}}}}}

\newcommand{\brackets}[2]{(#1{)#2}}
\newcommand{\evalprog}[3]{{\mathcal N}_p[\![#1{]\!]~#2{#3}}}
\newcommand{\eval}[4]{{\mathcal N}_e[\![#1{]\!]~#2{~#3{#4}}}}
\newcommand{\subst}[4]{#1{\{#3{ \mapsto #2{\}#4}}}}

\newcommand{\place}[3]{#2{\bullet #1{#3}}}
\newcommand{\ignore}[1]{}

\newcommand{\transformprog}[4]{{\mathcal T}_p^{#1}{[\![#2{]\!]~#3{#4}}}}
\newcommand{\transform}[7]{{\mathcal T}_e^{#1}{[\![#2{]\!]~#3{~#4{~#5{~#6{#7}}}}}}}
\newcommand{\transformcon}[7]{{\mathcal T}_\kappa^{#1}{[\![#2{]\!]~#3{~#4{~#5{~#6{#7}}}}}}}
\newcommand{\resprog}[3]{{\mathcal R}_p[\![#1{]\!]~#2{#3}}}
\newcommand{\resexp}[4]{{\mathcal R}_e[\![#1{]\!]~#2{~#3{#4}}}}

\title{The Next 700 Program Transformers}

\author{G.W. Hamilton\orcidID{000-0001-5954-6444}}

\institute{School of Computing, Dublin City University, Dublin, Ireland 
\email{geoffrey.hamilton@dcu.ie}}

\begin{document}

\maketitle

\begin{abstract}
In this paper, we describe a hierarchy of program transformers in which the transformer at each level of the hierarchy builds on top of those at lower levels. 
The program transformer at level 1 of the hierarchy corresponds to positive supercompilation, and that at level 2 corresponds to distillation. 
We prove that the transformers at each level terminate. We then consider the speedups that can be obtained at each level in the hierarchy, 
and try to characterise the improvements that can be made. 
\end{abstract}


\section{Introduction}

It is well known that programs written using functional programming languages often make use of intermediate data structures and thus can be inefficient. 
Several program transformation techniques have been proposed to eliminate some of these intermediate data structures; for example {\em partial evaluation} 
\cite{JONES93}, {\em deforestation} \cite{WADLER88} and {\em supercompilation} \cite{TURCHIN86}. {\em Positive supercompilation} \cite{SORENSEN96} 
is a variant of Turchin's supercompilation \cite{TURCHIN86} that was introduced in an attempt to study and explain the essentials of Turchin's supercompiler. Although  strictly 
more powerful than both partial evaluation and deforestation, S{\o}rensen has shown that positive supercompilation (without the identification of common 
sub-expressions in generalisation), and hence also partial evaluation and deforestation, can only produce a linear speedup in programs \cite{SORENSEN94A}.
Even with the identification of common sub-expressions in generalisation, superlinear speedups are obtained for very few interesting programs, and many obvious 
improvements cannot be made without the use of so-called `eureka' steps \cite{BURSTALL77}.
\begin{example}
\normalfont{Consider the function call $nrev~xs$ shown in Fig. \ref{example1}.
\begin{figure}[htb]
\begin{center}
\begin{tabular}{l}
$\expr{\where{\app{\var{nrev}}{\var{xs}}}
{\fundef{\app{\var{nrev}}{\var{xs}}}{\cas{\var{xs}}{\var{Nil}}{\var{Nil}}{\app{\app{\var{Cons}}{\var{x'}}}{\var{xs'}}}{\app{\app{\var{append}}{\brackets{\app{\var{nrev}}{\var{xs'}}}}}{\brackets{\app{\app{\var{Cons}}{\var{x'}}}{\var{Nil}}}}}} \\
\fundef{\app{\app{\var{append}}{\var{xs}}}{\var{ys}}}{\cas{\var{xs}}{\var{\var{Nil}}}{\var{ys}}{\app{\app{\var{Cons}}{\var{x'}}}{\var{xs'}}}{\app{\app{\var{Cons}}{\var{x'}}}{\brackets{\app{\app{\var{append}}{\var{xs'}}}{\var{ys}}}}}}}}$ \\
\\ 
$\expr{\where{\app{\var{qrev}}{\var{xs}}}
{\fundef{\app{\var{qrev}}{\var{xs}}}{\app{\app{\var{qrev'}}{\var{xs}}}{\var{Nil}}} \\
\fundef{\app{\app{\var{qrev'}}{\var{xs}}}{\var{ys}}}{\cas{\var{xs}}{\var{Nil}}{\var{ys}}{\app{\app{\var{Cons}}{\var{x'}}}{\var{xs'}}}{\app{\app{\var{qrev'}}{\var{xs'}}}{\brackets{\app{\app{\var{Cons}}{\var{x'}}}{\var{ys}}}}}}}}$
\end{tabular}
\end{center}
\caption{Alternative Definitions of List Reversal}
\label{example1}
\end{figure} 
This reverses the list $xs$, but the recursive function call $(nrev~xs')$ is an intermediate data structure, so 
in terms of time and space usage, it is quadratic with respect to the length of the list $xs$. A more efficient function 
that is linear with respect to the length of the list $xs$ is the function $qrev$ shown in Fig. \ref{example1}.

A number of algebraic transformations have been proposed that can perform this transformation
(e.g. \cite{WADLER87D}), making essential use of eureka steps requiring human insight and not easy to automate;
for the given example this can be achieved by appealing to a specific law stating the associativity of the $append$ function. 
However, none of the generic program transformation techniques mentioned above are capable of 
performing this transformation.}
\end{example}
The {\em distillation} algorithm \cite{HAMILTON07A,HAMILTON12A} was originally motivated by the need for automatic techniques that avoid the reliance 
on eureka steps to perform transformations such as the above. In positive supercompilation, generalisation and folding are performed only on expressions, 
while in distillation, generalisation and folding are also performed on recursive function representations ({\em process trees}). This allows a number of 
improvements to be obtained using distillation that cannot be obtained using positive supercompilation. 

The process trees that are generalised and folded in 
distillation are in fact those produced by positive supercompilation, so we can see that the definition of distillation is built on top of positive supercompilation. 
This suggests the existence of a hierarchy of program transformers, where the transformer at each level is built on top of those at lower levels, and more 
powerful transformations are obtained as we move up through this hierarchy. In this paper, we define such a hierarchy inductively, with positive supercompilation 
at level 1, distillation at level 2 and each new level defined in terms of the previous ones. Each of the transformers is capable of performing {\em fusion} to
eliminate intermediate data structures by fusing nested function calls. As we move up through the hierarchy, deeper nestings of function calls can be fused,
thus removing more intermediate data structures.

The remainder of this paper is structured as follows. In Section 2, we define the higher-order functional language on which the described transformations are performed. 
In Section 3, we give an overview of process trees and define a number of operations on them. In Section 4, we define the program transformer hierarchy, where the 
transformer at level 0 corresponds to the identity transformation, and each successive transformer is defined in terms of the previous ones. In Section 5, we prove that 
each of the transformers in our hierarchy terminates. In Section 6, we consider the efficiency improvements that can be obtained as we move up through this hierarchy. 
Section 7 concludes and considers related work.
\section{Language}
 \label{sec-language-definition}

In this section, we describe the call-by-name higher-order functional language that will be used throughout this paper. 
\begin{definition}[Language Syntax]
\normalfont{The syntax of this language is as shown in Fig. \ref{grammar}.} 
\end{definition}
\begin{figure}[htb]
\begin{center}
\begin{tabular}{@{\hspace*{0mm}}l@{\hspace*{1mm}}r@{\hspace*{1mm}}l@{\hspace*{1mm}}l@{\hspace*{0mm}}}
$\expr{\var{prog}}$ & ::= & $\expr{\wheredots{\var{e_{0}}}{\var{h_1}}{\var{e_1}}{\var{h_n}}{\var{e_n}}}$ & Program \\
\\
$\expr{\var{e}} \in Exp$ & ::= & $\expr{\var{x}}$ & Variable \\
& $|$ & $\expr{\app{\var{c}}{\args{\var{e_1}}{\var{e_n}}}}$ & Constructor Application \\
& $|$ & $\expr{\abs{\var{x}}{\var{e}}}$ & $\lambda$-Abstraction \\
& $|$ & $\expr{\var{f}}$ & Function Call \\
& $|$ & $\expr{\app{\var{e_0}}{\var{e_1}}}$ & Application \\
& $|$ & $\expr{\casedots{\var{e_0}}{\var{p_1}}{\var{e_1}}{\var{p_n}}{\var{e_n}}}$ & Case Expression \\ 
& $|$ & $\expr{\Letexp{\var{x}}{\var{e_0}}{\var{e_1}}}$ & Let Expression \\
\\
$\expr{\var{h}}$ & ::= & $\expr{\app{\var{f}}{\args{\var{x_1}}{\var{x_n}}}}$ & Function Header \\
\\
$\expr{\var{p}}$ & ::= & $\expr{\app{\var{c}}{\args{\var{x_1}}{\var{x_n}}}}$ & Pattern
\end{tabular}
\end{center}
\caption{Language Syntax}
\label{grammar}
\end{figure} 
Programs in the language consist of an expression to evaluate and a set of function definitions.
An expression can be a variable, constructor application, $\lambda$-abstraction, function call, application, {\bf case} or {\bf let}. 
Variables introduced by function definitions, $\lambda$-abstractions, {\bf case} patterns and {\bf let}s are {\em bound}; all other variables are {\em free}. 
We assume that bound variables are represented using De Bruijn indices. An expression that contains no free variables is said to be {\em closed}.
We write $e \equiv e'$ if $e$ and $e'$ differ only in the names of bound variables.

Each constructor has a fixed arity; for example $\expr{\var{Nil}}$ has arity 0 and $\expr{\var{Cons}}$ has arity 2. 
In an expression $\expr{\app{\var{c}}{\args{\var{e_{1}}}{\var{e_{n}}}}}$, $n$ must equal the arity of $c$. 
The patterns in {\bf case} expressions may not be nested.  No variable may appear more than once within a pattern. 
We assume that the patterns in a {\bf case} expression are non-overlapping and exhaustive.
It is also assumed that erroneous terms such as $\expr{\app{\brackets{\app{\var{c}}{\args{\var{e_1}}{\var{e_n}}}}}{\var{e}}}$ 
where $c$ is of arity $n$ and $\expr{\casedots{\brackets{\abs{x}{e}}}{\var{p_1}}{\var{e_1}}{\var{p_k}}{\var{e_k}}}$ cannot occur.
\begin{definition}[Substitution]
\normalfont{We use the notation $\theta = \{x_1 \mapsto e_1, \ldots, x_n \mapsto e_n\}$ to denote a {\em substitution}.
If $e$ is an expression, then $e\theta = e\{x_1 \mapsto e_1, \ldots, x_n \mapsto e_n\}$ is the result of simultaneously 
substituting the expressions $e_1,\ldots, e_n$ for the corresponding variables $x_1,\ldots,x_n$, respectively, 
in the expression $e$ while ensuring that bound variables are renamed appropriately to avoid name capture.
A {\em renaming} denoted by $\sigma$ is a substitution of the form $\{x_1 \mapsto x_1', \ldots, x_n \mapsto x_n'\}$.} 
\end{definition}
\begin{definition}[Shallow Reduction Context]
\normalfont{A shallow reduction context ${\cal C}$ is an expression containing a single hole $\bullet$ in the place of the redex, which can have one of the two following possible forms:
\begin{center}
${\cal C} ::= \bullet~e~|~{\bf case}~\bullet~{\bf of}~p_1 \Rightarrow e_1 \ldots p_n \Rightarrow e_n$
\end{center}
} 
\end{definition}
\begin{definition}[Evaluation Context]
\normalfont{An evaluation context ${\cal E}$ is represented as a sequence of shallow reduction contexts (known as a {\em zipper} 
\cite{HUET97}), representing the nesting of these contexts from innermost to outermost within which the redex is contained. 
An evaluation context can therefore have one of the two following possible forms:
\begin{center}
${\cal E} ::= \langle \rangle~|~\langle {\cal C}:{\cal E} \rangle$
\end{center}}
\end{definition}
\begin{definition}[Insertion into Evaluation Context]
\normalfont{The insertion of an expression $e$ into an evaluation context $\kappa$, denoted by $\expr{\place{e}{\kappa}}$, is defined as follows:
\begin{center}
\begin{tabular}{lcl}
$\expr{\place{e}{\langle \rangle}}$ & = & $e$ \\
$\expr{\place{e}{\langle (\bullet~e'):\kappa \rangle}}$ & = & $\expr{\place{(e~e')}{\kappa}}$ \\
\multicolumn{3}{l}{$\expr{\place{e}{\langle ({\bf case}~\bullet~{\bf of}~p_1 \Rightarrow e_1 \ldots p_n \Rightarrow e_n):\kappa \rangle}}$} \\
& = & $\expr{\place{({\bf case}~e~{\bf of}~p_1 \Rightarrow e_1 \ldots p_n \Rightarrow e_n)}{\kappa}}$
\end{tabular}
\end{center}} 
\end{definition}
\begin{figure}[b]
\begin{center}
\begin{tabular}{l}
$\expr{\evalprog{\var{e}}{\Delta}}$ = $\expr{\eval{\var{e}}{\langle \rangle}{\Delta}}$ \\
\\
$\expr{\eval{\app{\var{c}}{\args{\var{e_{1}}}{\var{e_{n}}}}}{\langle \rangle}{\Delta}}$ = $\expr{\app{\var{c}}{\args{\brackets{\eval{\var{e_1}}{\langle \rangle}{\Delta}}}{\brackets{\eval{\var{e_n}}{\langle \rangle}{\Delta}}}}}$ \\
$\expr{\eval{\app{\var{c}}{\args{\var{e_{1}}}{\var{e_{n}}}}}{\langle (\casedots{\bullet}{\var{p_{1}}}{\var{e_1'}}{\var{p_{k}}}{\var{e_k'}}):\kappa \rangle}{\Delta}}$ = \\
~~~~~~$\expr{\eval{\var{e'_i\{x_1 \mapsto e_1,\ldots,x_n \mapsto e_n\}}}{\kappa}{\Delta}}$ \\
~~~~~~where $\exists i \in \{1 \ldots k\}.\expr{\var{p_{i}}} = \expr{\app{\var{c}}{\args{\var{x_{1}}}{\var{x_{n}}}}}$ \\
$\expr{\eval{\abs{\var{x}}{\var{e}}}{\langle \rangle}{\Delta}}$ = $\lambda x.(\expr{\eval{\var{e}}{\langle \rangle}{\Delta}})$ \\
$\expr{\eval{\abs{\var{x}}{\var{e}}}{\langle (\expr{\app{\bullet}{\var{e'}}}):\kappa \rangle}{\Delta}}$ = $\expr{\eval{\var{e\{x \mapsto e'\}}}{\kappa}{\Delta}}$ \\
$\expr{\eval{f}{\kappa}{\Delta}}$ =  $\expr{\eval{\var{\lambda x_1 \ldots x_n.e}}{\kappa}{\Delta}}$ \\
~~~~~~where $(f~x_1 \ldots x_n=e )\in \Delta$ \\
$\expr{\eval{\app{\var{e_0}}{\var{e_1}}}{\kappa}{\Delta}}$ = $\expr{\eval{e_0}{\langle (\expr{\app{\bullet}{\var{e_1}}}):\kappa \rangle}{\Delta}}$ \\
$\expr{\eval{\casedots{\var{e_0}}{\var{p_{1}}}{\var{e_1}}{\var{p_{n}}}{\var{e_n}}}{\kappa}{\Delta}}$ = \\
~~~~~~$\expr{\eval{e_0}{\langle (\expr{\casedots{\bullet}{\var{p_{1}}}{\var{e_1}}{\var{p_{n}}}{\var{e_n}}}):\kappa \rangle}{\Delta}}$ \\
$\expr{\eval{\Letexp{\var{x}}{\var{e_0}}{\var{e_1}}}{\kappa}{\Delta}}$ = $\expr{\eval{\var{e_1\{x \mapsto e_0\}}}{\kappa}{\Delta}}$
\end{tabular}
\end{center}
\caption{Language Semantics}
\label{semantics}
\end{figure}
\begin{definition}[Language Semantics]
\normalfont{The normal order reduction semantics for programs in our language is defined by $\expr{\evalprog{\var{e}}{\Delta}}$ as shown in Fig. \ref{semantics}, 
where $e$ is the expression to be reduced (where it is assumed this contains no free variables) and $\Delta$ is the function environment.}
\end{definition}

Within the rules ${\cal N}_e$, $\kappa$ denotes the context of the expression under scrutiny. 
We always evaluate the redex of an expression first, with the remainder of the expression given by $\kappa$. 
\section{Process Trees}

The output of each of the transformers in our hierarchy are represented by {\em process trees}, as defined in \cite{SORENSEN94B}.
Within these process trees, the nodes are labelled with expressions. We write $e \rightarrow t_1, \ldots, t_n$ for a process tree with 
root node containing the expression $e$, where $t_1, \ldots, t_n$ are the sub-trees of this root node. We also write $e \rightarrow \epsilon$ 
for a process tree in which the root node has no sub-trees. We use $root(t)$ to denote the expression in the root node of process tree $t$.
Process trees may also contain three special kinds of node: 
\begin{itemize}
\item {\em Unfold nodes}: these are of the form $h \rightarrow t$, where $h$ is a function header and $t$ is the process tree resulting from transforming an expression after unfolding.
\item {\em Fold nodes}: these are of the form $h \rightarrow \epsilon$, where folding has been performed with respect to a previous unfold node
and the corresponding function headers are renamings of each other.
\item {\em Generalisation nodes}: these are of the form $t^x$, where the sub-tree $t$ has been generalised to variable $x$.
\end{itemize}

Within a sub-tree of a process tree, variables in unfold node function headers, $\lambda$-abstractions and {\bf case} expressions within ancestor nodes are bound; 
all other variables are free. We use $fv(t)$ and $bv(t)$ to denote the free variables (which includes generalisation variables) and bound variables of 
sub-tree $t$ respectively. We denote the application of a renaming $\sigma$ to a process tree $t$ by $t\sigma$, where the renaming $\sigma$ is applied to all 
the expressions in the nodes of the process tree.

When transforming an expression with a function in the redex, a level $k+1$ transformer will first transform the expression using a level $k$ transformer. 
The resulting process tree is then compared to previously encountered process trees generated at level $k$. If it is a renaming of a previous one, then 
folding is performed, and if it is an embedding of a previous one, then it is generalised. The use of process trees in this comparison allows us to abstract 
away from the number and order of the parameters in functions, and to focus on their recursive structure. We therefore define renaming, embedding and 
generalisation on process trees.
\begin{definition}[Process Tree Renaming]
\normalfont{Process tree $t$ is a {\em renaming} of process tree $t'$ if there is a renaming $\sigma$ and a relation ${\cal S} \subseteq Exp \times Exp$
between the expressions labelling the corresponding nodes of $t\sigma$ and $t'$ such that $t\sigma \cong t'$, where the relation $\cong$ is defined as follows:}
\end{definition}
\begin{center}
\begin{tabular}{ll}
1. & $(h \rightarrow t)\cong (h' \rightarrow t')$, if $(h,h') \in {\cal S} \wedge t \cong t'$ \\
2. & $(h \rightarrow \epsilon) \cong (h' \rightarrow \epsilon)$, if $\exists \sigma.(h\sigma,h'\sigma) \in {\cal S}$ \\
3. & $t_1^{x} \cong t_2^{x}$, if $t _1\cong t_2$ \\
4. & $(\phi(e_1,\ldots,e_n) \rightarrow t_1, \ldots, t_n) \cong (\phi(e_1',\ldots,e_n') \rightarrow t_1', \ldots, t_n')$, \\
& ~~~~if $((\phi(e_1,\ldots,e_n),(\phi(e_1',\ldots,e_n')) \in {\cal S} \wedge \forall i \in \{1 \ldots n\}.t_i \cong t_i'$
\end{tabular}
\end{center} 
The first rule is for unfold nodes, where the pair of function headers must belong to ${\cal S}$, and the corresponding 
sub-trees must be renamings. The second rule is for fold nodes, where the pair of function headers must be renamings 
of a pair of function headers belonging to ${\cal S}$ (from the corresponding unfold nodes). The third rule is for generalisation nodes,
where the corresponding generalised sub-trees must be renamings. The final rule is for all other nodes, where the pair of 
expressions in the corresponding root nodes must have the same top-level syntactic constructor, and the corresponding
sub-trees must also be renamings. This includes the pathological case where the nodes have no sub-trees (such as free 
variables which must have the same name, and bound variables which must have the same de Bruijn index).
\begin{definition}[Process Tree Embedding]
\normalfont{Process tree $t$ is {\em embedded} in process tree $t'$ if there is a renaming $\sigma$ and a relation ${\cal S} \subseteq Exp \times Exp$
between the expressions labelling the corresponding nodes of $t\sigma$ and $t'$ such that $t\sigma \trianglelefteq t'$, where the relation 
$\trianglelefteq$ is defined as follows:}
\end{definition}
\begin{center}
\begin{tabular}{ll}
1. & $(h \rightarrow t) \trianglelefteq (h' \rightarrow t')$, if $(h,h') \in {\cal S} \wedge t \trianglelefteq t'$ \\
2. & $(h \rightarrow \epsilon) \trianglelefteq (h' \rightarrow \epsilon)$, if $\exists \sigma.(h\sigma,h'\sigma) \in {\cal S}$ \\
3. & $t_1^{x} \trianglelefteq t_2^{x}$, if $t _1\trianglelefteq t_2$ \\
4. & $(\phi(e_1,\ldots,e_n) \rightarrow t_1, \ldots, t_n) \trianglelefteq (\phi(e_1',\ldots,e_n') \rightarrow t_1', \ldots, t_n')$, \\
& ~~~~if $((\phi(e_1,\ldots,e_n),(\phi(e_1',\ldots,e_n')) \in {\cal S} \wedge \forall i \in \{1 \ldots n\}.t_i \trianglelefteq t_i'$ \\
5. & $t \trianglelefteq (e \rightarrow t_1, \ldots, t_n)$, if $\exists i \in \{1 \ldots n\}.t \trianglelefteq t_i$ 
\end{tabular}
\end{center}
The first three rules are similar to those for the renaming relation $\cong$ for unfold, fold and generalisation nodes respectively.
The fourth rule is a {\em coupling} rule, where the pair of expressions in the root nodes must have the same top-level syntactic constructor
and the corresponding sub-trees of the root nodes must also be related to each other. This includes the pathological case where the root nodes 
have no sub-trees (such as free variables which must have the same name, and bound variables which must have the same de Bruijn index).
The final rule is a {\em diving} rule; this relates a process-tree with a sub-tree of a larger process tree. 
We write $t\preceq t'$ if $t \trianglelefteq t'$ and any rule other than the diving rule can be applied at the top level.
\begin{example}
\normalfont{Consider the two process trees in Fig. \ref{embedding} that correspond to the expressions $\expr{\app{\app{\var{append}}{\var{xs}}}{\brackets{\app{\app{\var{Cons}}{\var{x}}}{\var{Nil}}}}}$ and $\expr{\app{\app{\var{append}}{\brackets{\app{\app{\var{append}}{\var{xs}}}{\brackets{\app{\app{\var{Cons}}{\var{x}}}{\var{Nil}}}}}}}{(\app{\app{\var{Cons}}{\var{y}}}{}}}$ $Nil)$ respectively. 
Process tree (1) is embedded in process tree (2) by the relation $\preceq$.}
\end{example}
\begin{center}
\begin{figure}[htb]
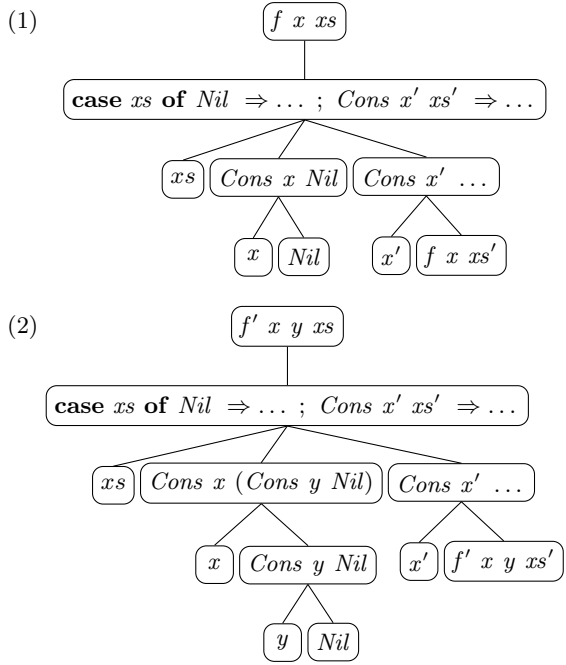

\begin{center}
\begin{tabular}{cc}
(1) & 
\Tree [.\node[draw, rounded corners,minimum size=5mm]{$\expr{\app{\app{\var{f}}{\var{x}}}{\var{xs}}}$}; 
            [.\node[draw, rounded corners,minimum size=5mm]{$\expr{\Cas{\var{xs}}{\var{Nil}}{\ldots}{\app{\app{\var{Cons}}{\var{x'}}}{\var{xs'}}}{\ldots}}$};
               [.\node[draw, rounded corners,minimum size=5mm]{\raisebox{1pt}{$xs$}}; ]
               [.\node[draw, rounded corners,minimum size=5mm]{$\expr{\app{\app{\var{Cons}}{\var{x}}}{\var{Nil}}}$};
                  [.\node[draw, rounded corners,minimum size=5mm]{\raisebox{2pt}{$\expr{\var{x}}$}};]
                  [.\node[draw, rounded corners,minimum size=5mm]{$\expr{\var{Nil}}$};] ]
               [.\node[draw, rounded corners,minimum size=5mm]{$\expr{\app{\app{\var{Cons}}{\var{x'}}}{\ldots}}$}; 
                  [.\node[draw, rounded corners,minimum size=5mm]{$\expr{\var{x'}}$};]
                  [.\node[draw, rounded corners,minimum size=5mm]{\raisebox{1pt}{$\expr{\app{\app{\var{f}}{\var{x}}}{\var{xs'}}}$}};] ] ] ] \\
\\
(2) &
\Tree [.\node[draw, rounded corners,minimum size=5mm]{$\expr{\app{\app{\app{\var{f'}}{\var{x}}}{\var{y}}}{\var{xs}}}$}; 
            [.\node[draw, rounded corners,minimum size=5mm]{$\expr{\Cas{\var{xs}}{\var{Nil}}{\ldots}{\app{\app{\var{Cons}}{\var{x'}}}{\var{xs'}}}{\ldots}}$};
               [.\node[draw, rounded corners,minimum size=5mm]{\raisebox{1pt}{$xs$}}; ]
               [.\node[draw, rounded corners,minimum size=5mm]{\raisebox{1pt}{$\expr{\app{\app{\var{Cons}}{\var{x}}}{\brackets{\app{\app{\var{Cons}}{\var{y}}}{\var{Nil}}}}}$}};
                  [.\node[draw, rounded corners,minimum size=5mm]{\raisebox{1pt}{$\expr{\var{x}}$}};]
                  [.\node[draw, rounded corners,minimum size=5mm]{$\expr{\app{\app{\var{Cons}}{\var{y}}}{\var{Nil}}}$};
                     [.\node[draw, rounded corners,minimum size=5mm]{\raisebox{2pt}{$\expr{\var{y}}$}};]
                     [.\node[draw, rounded corners,minimum size=5mm]{$\expr{\var{Nil}}$};] ] ]
               [.\node[draw, rounded corners,minimum size=5mm]{$\expr{\app{\app{\var{Cons}}{\var{x'}}}{\ldots}}$}; 
                  [.\node[draw, rounded corners,minimum size=5mm,distance=1.3cm]{$\expr{\var{x'}}$};]
                  [.\node[draw, rounded corners,minimum size=5mm]{\raisebox{1pt}{$\expr{\app{\app{\app{\var{f'}}{\var{x}}}{\var{y}}}{\var{xs'}}}$}};] ] ] ] \\
\end{tabular}
\end{center}
\caption{Embedded Process Trees}
\label{embedding}
\end{figure}
\end{center}
The generalisation of a process tree involves replacing sub-trees with generalisation variables and creating {\em tree substitutions}. 
\begin{definition}[Tree Substitution]
\normalfont{We use the notation $\varphi = \{x_1 \mapsto t_1, \ldots,$ $x_n \mapsto t_n\}$ to denote a {\em tree substitution}.
If $t$ is an process tree, then $t\varphi = t\{x_1 \mapsto t_1, \ldots, x_n \mapsto t_n\}$ is the result of simultaneously 
substituting the sub-trees $t_1,\ldots, t_n$ for the corresponding variables $x_1,\ldots,x_n$, respectively, 
in the process tree $t$ while ensuring that bound variables are renamed appropriately to avoid name capture.} 
\end{definition}
\begin{definition}[Generalisation]
\normalfont{The generalisation of two process trees $t$ and $t'$ is a triple $(t_g,\varphi_1,\varphi_2)$ where $\varphi_1$ and $\varphi_2$ 
are tree substitutions such that $t_g \varphi_1 \cong t$ and $t_g \varphi_2 \cong t'$.}
\end{definition}
\begin{definition}[The Generalisation Operator $\sqcap$]
\normalfont{The generalisation of two process trees $t$ and $t'$, where $t$ and $t'$ are related by $\preceq$, is given by $t \sqcap t'$.
The following rewrite rules are repeatedly applied to the initial triple $(x,\{x \mapsto t\},\{x \mapsto t'\})$, while 
the process trees associated with the same variable in each of the tree substitutions are related by $\preceq$: 
\begin{center}
$\begin{array}{ccc}
(t_g,\{x \mapsto t_1^{x'}\},\{x \mapsto t_2^{x'}\}) & \Rightarrow & (t_g,\{x \mapsto t_1\},\{x \mapsto t_2\}) \\
\\
\left(\begin{array}{c}
t_g, \\
\{x \mapsto (e \rightarrow t_1, \ldots, t_n)\} \cup \varphi, \\
\{x \mapsto (e' \rightarrow t_1', \ldots, t_n')\} \cup \varphi'
\end{array}\right) 
& \Rightarrow 
& \left(\begin{array}{c}
t_g\{x \mapsto (e \rightarrow x_1, \ldots, x_n)\}, \\
\{x_1 \mapsto t_1,\ldots,x_n \mapsto t_n\} \cup \varphi, \\
\{x_1 \mapsto t_1',\ldots,x_n \mapsto t_n'\} \cup \varphi'
\end{array}\right)
\end{array}$
\end{center}
In the first rule, if both the process trees related by $\preceq$ have a generalisation node at the root, then this generalisation is removed.
The second rule adds the root node of one of the process trees that is related by $\preceq$ into the generalised tree, and new generalisation variables are 
added for the corresponding sub-trees of these root nodes. Note that it does not matter which of the two input process trees the expressions 
in the resulting generalised process tree come from, so long as they all come from one of them (so the corresponding unfold and fold nodes still match); 
the resulting generalised program will be the same.

The following rewrite rule is then exhaustively applied to the triple resulting from the above rewrites to identify common substitutions that were previously given different names:
\begin{center}
$\begin{array}{ccc}
\left(\begin{array}{c}
t_g, \\
\{x \mapsto t,x' \mapsto t\} \cup \varphi, \\
\{x \mapsto t',x' \mapsto t'\} \cup \varphi'
\end{array}\right) 
& \Rightarrow 
& \left(\begin{array}{c}
t_g\{x \mapsto x'\}, \\
\{x' \mapsto t\} \cup \varphi, \\
\{x' \mapsto t'\} \cup \varphi'
\end{array}\right)
\end{array}$
\end{center}}
\end{definition}
\begin{example}
\normalfont{The result of generalising the two process trees in Fig. \ref{embedding} is shown in Fig. \ref{generalisation}, with the mismatched nodes replaced
by the generalisation variable $v$.}
\end{example}
\begin{center}
\begin{figure}[htb]
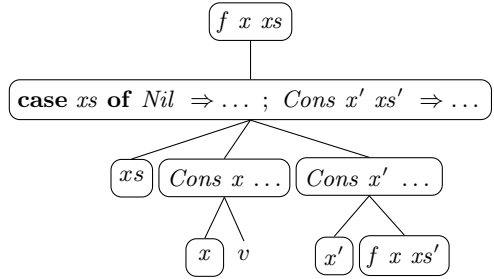

\begin{center}
\Tree [.\node[draw, rounded corners,minimum size=5mm]{$\expr{\app{\app{\var{f}}{\var{x}}}{\var{xs}}}$}; 
            [.\node[draw, rounded corners,minimum size=5mm]{$\expr{\Cas{\var{xs}}{\var{Nil}}{\ldots}{\app{\app{\var{Cons}}{\var{x'}}}{\var{xs'}}}{\ldots}}$};
               [.\node[draw, rounded corners,minimum size=5mm]{\raisebox{2pt}{$xs$}}; ]
               [.\node[draw, rounded corners,minimum size=5mm]{$\expr{\app{\app{\var{Cons}}{\var{x}}}{\var{\ldots}}}$};
                  [.\node[draw, rounded corners,minimum size=5mm]{\raisebox{2pt}{$\expr{\var{x}}$}};]
                  [.\node{\raisebox{2pt}{$v$}};] ]
               [.\node[draw, rounded corners,minimum size=5mm]{$\expr{\app{\app{\var{Cons}}{\var{x'}}}{\ldots}}$}; 
                  [.\node[draw, rounded corners,minimum size=5mm]{$\expr{\var{x'}}$};]
                  [.\node[draw, rounded corners,minimum size=5mm]{\raisebox{1pt}{$\expr{\app{\app{\var{f}}{\var{x}}}{\var{xs'}}}$}};] ] ] ] 
\end{center}
\caption{Generalised Process Tree}
\label{generalisation}
\end{figure}
\end{center}
\begin{definition}[The Generalisation Substitution Operator $\downarrow$]
\normalfont{The operation $\downarrow$ is applied to the triple resulting from generalisation to make the second tree substitution in the triple explicit within the generalised process tree:
\begin{center}
$\begin{array}{c}
\downarrow (t_g,\varphi,\{x_1 \mapsto t_1,\ldots,x_n \mapsto t_n\}) = t_g\{x_1 \mapsto t_1^{x_1}, \ldots, x_n \mapsto t_n^{x_n}\} 
\end{array}$
\end{center}}
\end{definition}
We now show how a program can be residualised from a process tree. 
\begin{definition}[Residualisation]
\normalfont{A program can be residualised from a process tree $t$ as $\expr{\resprog{\var{t}}{\Delta}}$ (where $\Delta$ is the set of previous function definitions) 
using the rules as shown in Fig. \ref{residual}.} 
\end{definition}
\begin{center}
\begin{figure}[htb]
\begin{center}
\begin{tabular}[t]{l}
$\expr{\resprog{\var{t}}{\Delta}}$ = $(\uparrow (e,\theta),\Delta')$ \\
~~~~~~where $(e,\theta,\Delta') = \expr{\resexp{t}{\{\}}{\Delta}}$ \\
\\
$\expr{\resexp{\var{h \rightarrow t}}{\rho}{\Delta}}$ = $(h',\theta,\{h' = e\} \cup \Delta')$ \\
~~~~~~where $\begin{array}[t]{l}
\expr{\resexp{\var{t}}{(\rho \cup \{h = h'\})}{\Delta}} = (e,\theta,\Delta') \\
h'=f~x_1 \ldots x_n~(f$ is fresh$, \{x_1 \ldots x_n\} = fv(t))
\end{array}$ \\
$\expr{\resexp{\var{h \rightarrow \epsilon}}{\rho}{\Delta}}$ = $\left\{\begin{array}{ll}
(h''\sigma,\{\},\{\}), & $if $\exists h'.(h'=h'') \in \rho \wedge h \equiv h'\sigma \\
(h,\{\},\{h=e\}), & $otherwise (where $(h=e) \in \Delta)
\end{array}\right.$ \\
$\expr{\resexp{\var{e \rightarrow \epsilon}}{\rho}{\Delta}} = (e,\{\},\{\})$ \\
$\expr{\resexp{\var{(c~e_1 \ldots e_n) \rightarrow t_1,\ldots,t_n}}{\rho}{\Delta}} = (c~e_1' \ldots e_n',\bigcup\limits_{i=1}^{n} \theta_i,\bigcup\limits_{i=1}^{n} \Delta_i)$ \\
~~~~~~where $\forall i \in \{1 \ldots n\}.\expr{\resexp{\var{t_i}}{\rho}{\Delta}} = (e_i',\theta_i,\Delta_i)$ \\
$\expr{\resexp{\var{(\lambda x.e) \rightarrow t}}{\rho}{\Delta}} = (\expr{\abs{\var{x}}{\var{e'}}},\theta,\Delta')$ \\
~~~~~~where $\expr{\resexp{\var{t}}{\rho}{\Delta}} = (e',\theta,\Delta')$ \\
$\expr{\resexp{\var{(e_0~e_1) \rightarrow t_0,t_1}}{\rho}{\Delta}} = (e_0'~e_1',\bigcup\limits_{i=1}^{2} \theta_i,\bigcup\limits_{i=1}^{2}\Delta_i)$ \\
~~~~~~where $\forall i \in \{0 \ldots 1\}.\expr{\resexp{\var{t_i}}{\rho}{\Delta}} = (e_i',\theta_i,\Delta_i)$ \\
$\expr{\resexp{\var{({\bf case}~e_0~{\bf of}~p_1 \Rightarrow e_1 \ldots p_n \Rightarrow e_n) \rightarrow t_0,\ldots,t_n}}{\rho}{\Delta}}$ = \\
~~~~~~$(\expr{\casedots{\var{e_0'}}{\var{p_{1}}}{\var{e_1'}}{\var{p_{n}}}{\var{e_n'}}},\bigcup\limits_{i=0}^{n} \theta_i,\bigcup\limits_{i=0}^{n} \Delta_i)$ \\
~~~~~~where $\forall i \in \{0 \ldots n\}.\expr{\resexp{\var{t_i}}{\rho}{\Delta}} = (e_i',\theta_i,\Delta_i)$ \\
$\expr{\resexp{\var{t^x}}{\rho}{\Delta}}$ = $(x~x_1 \ldots x_n,\{x \mapsto \lambda x_1 \ldots x_n.e\} \cup \theta,\Delta')$ \\
~~~~~~where $\expr{\resexp{\var{t}}{\rho}{\Delta}} = (e,\theta,\Delta')$ ($\{x_1 \ldots x_n\} = bv(t)$)
\end{tabular}
\end{center}
\caption{Rules For Residualisation}
\label{residual}
\end{figure}
\end{center}
Within the rules ${\cal R}_e$, the parameter $\rho$ contains the unfold node function headers and the associated new function headers that are 
created for them in the residualised program.
On encountering an unfold node, a new function header is created, associated with the unfold node function header, and added to $\rho$. 
Note that this new function header may not have the same variables as the one in the unfold node, as new variables may have been added to the sub-tree 
as a result of generalisation. On encountering a fold node, if it matches a corresponding unfold node, then a recursive call of the function associated 
with the unfold node function header in $\rho$ is created. Otherwise, the fold node function call and its previous definition are used (this will only occur
for process trees produced by our level 0 transformer).
Within the rules ${\cal R}_e$, environments $\theta$ and $\Delta$ are returned in addition to the residual expression.
$\theta$ contains the generalisation variables and their associated extracted values. $\Delta$ contains the set of newly created function definitions. 
On encountering a generalisation node, the generalised expression is extracted, added to $\theta$, and replaced with an application of the generalisation variable.
The extracted expression is abstracted over its bound variables so that these are not extracted outside their binder. The set of abstracted variables will also be the 
arguments in the generalised variable application.
The generalisation environment $\theta$ is converted to {\bf let}s at the top level using the generalisation extraction operator $\uparrow$, which is defined as follows.
\begin{definition}[The Generalisation Extraction Operator $\uparrow$]
\normalfont{The generalisation extraction operator $\uparrow$, where $\uparrow(e,\theta)$ extracts the generalisation environment $\theta$ from the 
expression $e$ and is defined as follows:
\begin{center}
$\uparrow(e,\{x_1 \mapsto e_1,\ldots,x_n \mapsto e_n\}) = \expr{\multiletexp{\var{x_1}}{\var{e_1}}{\var{x_n}}{\var{e_n}}{\var{e}}}$
\end{center}}
\end{definition}

\section{A Hierarchy of Program Transformers}

In this section, we define our hierarchy of transformers. The level $k$ transformer is defined as $\expr{\transformprog{k}{\var{e}}{\Delta}}$,
where $e$ is the expression to be transformed and $\Delta$ is the function environment. It is assumed that the input program contains no 
$\lambda$-abstractions or {\bf let} expressions ($\lambda$-abstractions can be replaced by named functions and {\bf let} expressions can be substituted).
The output of the transformer is a process tree from which the transformed program can be residualised.

\subsection{Level 0 Transformer}

Level 0 in our hierarchy essentially corresponds to the identity transformation, and just converts a term to a corresponding process tree:
\begin{center}
$\expr{\transformprog{0}{\var{\phi(e_1,\ldots,e_n)}}{\Delta}} = \phi(e_1,\ldots,e_n) \rightarrow (\expr{\transformprog{0}{\var{e_1}}{\Delta}}),\ldots,(\expr{\transformprog{0}{\var{e_n}}{\Delta}})$
\end{center}
The sub-terms of the original term are therefore simply mapped to the corresponding sub-trees of the resulting process tree.

\subsection{Level $k+1$ Transformers}

Each subsequent level ($k+1$) in our hierarchy is built on top of the previous levels. The rules for level $k+1$ transformation of programs in our language are
defined by $\expr{\transformprog{k+1}{\var{e}}{\Delta}}$ as shown in Fig. \ref{levelk+1rules}.
\begin{figure}[htb]
\begin{center}
\begin{tabular}{ll}
(1) & $\expr{\transformprog{k+1}{\var{e}}{\Delta}}$ = $\expr{\transform{k+1}{e}{\langle \rangle}{\{\}}{\{\}}{\Delta}}$ \\
\\
(2) & $\expr{\transform{k+1}{\var{x~x_1 \ldots x_n}}{\kappa}{\rho}{\theta}{\Delta}}$ = $\begin{array}[t]{l}
\expr{\transform{\kappa}{\var{t^x}}{\kappa}{\rho}{\theta}{\Delta}}, $ if $x \in dom(\theta) \\
$where $t = \expr{\transform{k+1}{\var{\theta(x)~x_1 \ldots x_n}}{\kappa}{\rho}{\theta}{\Delta}}
\end{array}$ \\
(3) & $\expr{\transform{k+1}{\var{x}}{\langle \rangle}{\rho}{\theta}{\Delta}}$ = $x \rightarrow \epsilon$ \\
(4) & $\expr{\transform{k+1}{\var{x}}{(\langle (\casedots{\bullet}{\var{p_1}}{\var{e_1}}{\var{p_n}}{\var{e_n}}):\kappa \rangle)}{\rho}{\theta}{\Delta}}$ = \\ 
\multicolumn{2}{l}{$\expr{\transformcon{k+1}{x \rightarrow \epsilon}{\langle (\casedots{\bullet}{p_1}{\subst{\brackets{\place{e_1}{\kappa}}}{p_1}{x}}{p_n}{\subst{\brackets{\place{e_n}{\kappa}}}{p_n}{x}}):\langle \rangle \rangle}{\rho}{\theta}{\Delta}}$} \\
(5) & $\expr{\transform{k+1}{\var{x}}{\langle (\expr{\app{\bullet}{\var{e}}}):\kappa \rangle}{\rho}{\theta}{\Delta}}$ = $\expr{\transformcon{k+1}{x \rightarrow \epsilon}{\langle (\expr{\app{\bullet}{\var{e}}}):\kappa \rangle}{\rho}{\theta}{\Delta}}$ \\
(6) & $\expr{\transform{k+1}{\app{\var{c}}{\args{\var{e_{1}}}{\var{e_{n}}}}}{\langle \rangle}{\rho}{\theta}{\Delta}}$ = \\
& ~~~~~~$(\expr{\app{\var{c}}{\args{\var{e_{1}}}{\var{e_{n}}}}}) \rightarrow (\expr{ \transform{k+1}{\var{e_1}}{\langle \rangle}{\rho}{\theta}{\Delta}}),\ldots,(\expr{ \transform{k+1}{\var{e_n}}{\langle \rangle}{\rho}{\theta}{\Delta}})$ \\
(7) & $\expr{\transform{k+1}{\app{\var{c}}{\args{\var{e_{1}}}{\var{e_{n}}}}}{\langle (\casedots{\bullet}{\var{p_{1}}}{\var{e_1'}}{\var{p_{k}}}{\var{e_k'}}):\kappa \rangle}{\rho}{\theta}{\Delta}}$ = \\
& ~~~~~~$\expr{ \transform{k+1}{\var{e'_i\{x_1 \mapsto e_1,\ldots,x_n \mapsto e_n\}}}{\kappa}{\rho}{\theta}{\Delta}}$ \\
& ~~~~~~where $\exists i \in \{1 \ldots k\}.\expr{\var{p_{i}}} = \expr{\app{\var{c}}{\args{\var{x_{1}}}{\var{x_{n}}}}}$ \\
(8) & $\expr{\transform{k+1}{\abs{\var{x}}{\var{e_0}}}{\langle \rangle}{\rho}{\theta}{\Delta}} = (\expr{\abs{\var{x}}{\var{e_0}}}) \rightarrow (\expr{ \transform{k+1}{\var{e_0}}{\langle \rangle}{\rho}{\theta}{\Delta}})$ \\
(9) & $\expr{\transform{k+1}{\abs{\var{x}}{\var{e_0}}}{\langle (\expr{\app{\bullet}{\var{e_1}}}):\kappa \rangle}{\rho}{\theta}{\Delta}} = \expr{ \transform{k+1}{\var{e_0\{x \mapsto e_1\}}}{\kappa}{\rho}{\theta}{\Delta}}$ \\
(10) & $\expr{\transform{k+1}{\var{f}}{\kappa}{\rho}{\theta}{\Delta}} = \begin{array}[t]{l}\left\{\begin{array}[m]{l}
h\sigma \rightarrow \epsilon,\hfill $ if $\exists (h=t') \in \rho, \sigma.t'\sigma \cong t \\
\expr{\transform{k+1}{\var{e}}{\langle \rangle}{\rho}{\theta}{\Delta'}}, \hfill $ if $\exists (h=t') \in \rho, \sigma.t' \sigma \preceq t \\
$where $\begin{array}[t]{l}
(e,\Delta') = \expr{\resprog{\var{\downarrow(t' \sigma \sqcap t)}}{\Delta}}
\end{array} \\
\expr{h' \rightarrow \transform{k+1}{\var{\lambda x_1 \ldots x_n.e}}{\kappa}{(\rho \cup \{h'=t\})}{\theta}{\Delta}}, \hfill $otherwise$ \\
$where $\begin{array}[t]{l}
(f~x_1 \ldots x_n=e) \in \Delta \\
h'=f'~x_1' \ldots x_k'$ ($f'$ is fresh, $\{x_1' \ldots x_k'\} = fv(t))
\end{array}
\end{array}\right. \\
$where $t = \expr{\transformprog{k}{\place{\var{f}}{\kappa}}{\Delta}}
\end{array}$ \\
(11) & $\expr{\transform{k+1}{\app{\var{e_0}}{\var{e_1}}}{\kappa}{\rho}{\theta}{\Delta}}$ = $\expr{\transform{k+1}{\var{e_0}}{\langle (\expr{\app{\bullet}{\var{e_1}}}):\kappa \rangle}{\rho}{\theta}{\Delta}}$ \\
(12) & $\expr{\transform{k+1}{\casedots{\var{e_0}}{\var{p_{1}}}{\var{e_1}}{\var{p_{n}}}{\var{e_n}}}{\kappa}{\rho}{\theta}{\Delta}}$ = \\
& ~~~~~~$\expr{\transform{k+1}{\var{e_0}}{\langle (\expr{\casedots{\bullet}{\var{p_{1}}}{\var{e_1}}{\var{p_{n}}}{\var{e_n}}}):\kappa \rangle}{\rho}{\theta}{\Delta}}$ \\
(13) & $\expr{\transform{k+1}{\Letexp{\var{x}}{\var{e_0}}{\var{e_1}}}{\kappa}{\rho}{\theta}{\Delta}}$ = $\expr{\transform{k+1}{\var{e_1}}{\kappa}{\rho}{(\theta \cup \{x \mapsto e_0\})}{\Delta}}$ \\
\\
(14) & $\expr{\transformcon{k+1}{\var{t}}{\langle \rangle}{\rho}{\theta}{\Delta}}$ = $t$ \\
(15) & $\expr{\transformcon{k+1}{\var{t}}{(\kappa = \langle (\bullet~e):\kappa' \rangle)}{\rho}{\theta}{\Delta}}$ = \\
& ~~~~~~$\expr{\transformcon{k+1}{(\expr{\place{root(t)}{\kappa}}) \rightarrow t,(\expr{\transform{k+1}{e}{\langle \rangle}{\rho}{\theta}{\Delta}})}{\kappa'}{\rho}{\theta}{\Delta}}$ \\
(16) & $\expr{\transformcon{k+1}{\var{t}}{(\kappa = \langle (\casedots{\bullet}{\var{p_{1}}}{\var{e_1}}{\var{p_{n}}}{\var{e_n}}):\kappa' \rangle)}{\rho}{\theta}{\Delta}}$ = \\
& ~~~~~~$(\expr{\place{root(t)}{\kappa}}) \rightarrow t,(\expr{ \transform{k+1}{e_1}{\kappa'}{\rho}{\theta}{\Delta}}),\ldots,(\expr{ \transform{k+1}{e_n}{\kappa'}{\rho}{\theta}{\Delta}})$
\end{tabular}
\end{center}
\caption{Level $k+1$ Transformation Rules}
\label{levelk+1rules}
\end{figure}
Within the rules ${\cal T}_e^{k+1}$, $\kappa$ denotes the context of the expression under scrutiny, $\rho$ contains memoised process trees and their associated new
function headers, and $\theta$ gives the {\bf let} variables and their associated values. For most of the level $k+1$ transformation rules, normal order reduction is 
applied to the current term, as for the semantics given in Fig. \ref{semantics}. 

In rule (10), if the redex of the current term is a function, then it is transformed by the 
transformer one level lower in the hierarchy (level $k$) producing a process tree; this is therefore where the transformer builds on all the transformers at lower levels.
This level $k$ process tree is compared to the previous process trees produced at level $k$ (contained in $\rho$). If the process tree is a {\em renaming} of a 
previous one, then {\em folding} is performed, and a fold node is created using a recursive call of the function associated with the previous process tree in $\rho$. 
If the process tree is an {\em embedding} of a previous one, then {\em generalisation} is performed;
the result of this generalisation is then residualised and further transformed. Otherwise, the current process tree is memoised by being associated with a new function 
call in $\rho$; an unfold node is created with this new function call in the root node, with the result of transforming the unfolding of the current term as its 
sub-tree. 

In rule (7), if the context surrounding a constructor application redex is a {\bf case}, then pattern matching is performed and the appropriate branch of the
{\bf case} is selected, thus removing the constructor application. This is where our transformers actually remove intermediate data structures.

In rule (4), if the context surrounding a variable redex is a {\bf case}, then information is propagated to each branch of the {\bf case}
to indicate that this variable has the value of the corresponding branch pattern.

In rule (13), a {\bf let} variable is associated with its value in the environment $\theta$ and the {\bf let} is removed. 
In rule (2), if one of these {\bf let} variables is subsequently encountered, then it is replaced with a generalisation node; 
the original term generalisation is therefore replaced by a tree generalisation.

The rules ${\cal T}_\kappa^{k+1}$ are defined on a process tree and a surrounding context. 
These rules are applied when the normal-order reduction of the input program becomes `stuck' as a result of
encountering a variable in the redex position. In this case, the surrounding context is further transformed.

\section{Termination}

In order to prove that each of the transformers in our hierarchy terminate, we need to show that in any infinite 
sequence of process trees encountered during transformation $t_0, t_1, \ldots$ there definitely exists some 
$i < j$ where $t_i \preceq t_j$, so an embedding must eventually be encountered and transformation will 
not continue indefinitely without folding or generalising.  This amounts to proving that the embedding relation 
$\preceq$ is a {\em well-quasi order}. 
\begin{lemma}[$\preceq$ is a Well-Quasi Order]
\normalfont{The embedding relation $\preceq$ is a {\em well-quasi order} on any sequence of process trees 
that are encountered during transformation at level $k>0$ in our hierarchy.}
\label{wellquasi}
\end{lemma}
\begin{proof}
\normalfont{The proof is by induction on the hierarchy level $k$. 

For level 1, the proof is similar to that given in \cite{KLYUCHNIKOV10A}. This involves showing that there are a finite number of 
functors (function names and constructors) in the language. The process trees encountered during transformation are those produced at level 0, 
so the function names will be those from the original program, so must be finite. 
Applications of different arities are replaced with separate constructors; we prove that arities are bounded, so there are a finite number of these. 
We also replace {\bf case} expressions with constructors. Since bound variables are defined using de Bruijn indices, each of these are replaced 
with separate constructors; we also prove that de Bruijn indices are bounded. The overall number of functors is therefore finite, 
so Kruskal's tree theorem can then be applied to show that $\preceq$ is a well-quasi-order at level 1 in our hierarchy.

At level $k+1$, the process trees encountered during transformation are those produced at level $k$ and must be finite (by the inductive hypothesis).
The number of functions in these process trees must therefore be finite, and the same argument given above for level 1 also applies here, so $\preceq$ 
is a well-quasi-order at level $k+1$ in our hierarchy.}
\end{proof}
Since we only check for embeddings for process trees resulting from the transformation of expressions which have a named function as redex, 
we need to show that every potentially infinite sequence of expressions encountered during transformation must include expressions of this form.
\begin{lemma}[Function Unfolding During Transformation]
\normalfont{Every infinite sequence of transformation steps must include function unfolding.}
\end{lemma}
\begin{proof}
\normalfont{Every infinite sequence of transformation steps must include either function unfolding or 
$\lambda$-application. Since we do not allow $\lambda$-abstractions in our input program, the only way 
in which new $\lambda$-abstractions can be introduced is by function unfolding. 
Thus, every infinite sequence of transformation steps must include function unfolding.}
\end{proof}
\begin{theorem}[Termination of Transformation]
\normalfont{The transformation algorithm always terminates.}
\end{theorem}
\begin{proof}
\normalfont{
The proof is by contradiction. If the transformation algorithm did not terminate, then the set of memoised 
process trees in $\rho$ must be infinite. Every new process tree which is added to $\rho$ cannot have any of the previous 
process trees in $\rho$ embedded within it by the homeomorphic embedding relation $\preceq$, since 
generalisation would have been performed instead. However, this contradicts the fact that $\preceq$ is a 
well-quasi-order (Lemma \ref{wellquasi}).}
\end{proof}

\section{Speedups}

In this section, we look at the efficiency gains that can be obtained at different levels in our program transformation hierarchy.
\begin{theorem}[Exponential Speedups]
\normalfont{Exponential speedups can only be obtained above level 0 in our hierarchy if common sub-expression
elimination is performed during generalisation.}
\end{theorem}
\begin{proof}
An exponential speedup can only be obtained if the number of recursive calls of a function is reduced. This can only happen if some of these recursive calls
are identified by common sub-expression elimination.
\end{proof}
\begin{example}
Consider the following program from \cite{GLUCK16}:
\begin{center}
$\expr{\where{\app{\app{\var{f}}{\var{x}}}{\var{x}}}{\fundef{\app{\app{\var{f}}{\var{x}}}{\var{y}}}{\cas{\var{x}}{\var{Zero}}{\var{y}}{\var{Succ(x)}}{\app{\app{\var{f}}{\brackets{\app{\app{\var{f}}{\var{x}}}{\var{x}}}}}{\brackets{\app{\app{\var{f}}{\var{x}}}{\var{x}}}}}}}}$
\end{center}
This program takes exponential time $O(2^n)$, where $n$ is the size of the variable $x$.
During transformation at level 1 in our hierarchy, the process tree corresponding to $(f~x~x)$ is extracted twice, but then identified by common sub-expression
elimination to obtain the following program:
\begin{center}
$\expr{\where{\app{\var{f'}}{\var{x}}}{\fundef{\app{\var{f'}}{\var{x}}}{\cas{\var{x}}{\var{Zero}}{\var{Zero}}{\var{Succ(x)}}{\app{\var{f'}}{\brackets{\app{\var{f'}}{\var{x}}}}}}}}$
\end{center}
This program takes linear time $O(n)$ on the same input, so an exponential speedup has been achieved. In practice we have found that such exponential
improvements are obtained for very few useful programs; it is very unlikely that a programmer would write such an inefficient program when a much better
solution exists.
\end{example}
We now look at the improvements in efficiency that can be obtained without common sub-expression elimination. 
\begin{theorem}[Non-Exponential Speedups]
\normalfont{Without the use of common sub-expression elimination, the maximum speedup factor possible at level $k>0$ in our hierarchy for input of size $n$
 is $O(n^{k-1})$.}
\end{theorem}
\begin{proof}
The proof is by induction on the hierarchy level $k$. For level 1, the proof is as given in \cite{SORENSEN94A};
since there can only be a constant number of reduction steps removed between each successive call of a function, at most a linear speedup is possible.
For level $k+1$, there will be a constant number of calls to functions that were transformed at level $k$ between each successive call of a level $k+1$ function. 
By the inductive hypothesis, the maximum speedup factor for each level $k$ function is $O(n^{k-1})$, so the maximum speedup factor at level $k+1$ is $O(n^k)$.
\end{proof}
\begin{example}
Consider the transformation of the na\"ive reverse program shown in Fig. \ref{example1}, which has $O(n^2)$ runtime. If this program is transformed
at level 1 in our hierarchy, then no improvements are obtained. However, if we transform this program at level 2 in our hierarchy, we end up having to
transform a term equivalent to the following at level 1:
\begin{center}
$\expr{\app{\app{\var{append}}{\brackets{\app{\app{\var{append}}{\var{xs}}}{\brackets{\app{\app{\var{Cons}}{\var{x}}}{\var{Nil}}}}}}}{\brackets{\app{\app{\var{Cons}}{\var{y}}}{\var{Nil}}}}}$
\end{center}
Within this term, the list $xs$ has to be traversed twice. This term is transformed to one equivalent to the following at level 1 (process tree (2) in Fig. \ref{embedding} is
the process tree produced as a result of this transformation):
\begin{center}
$\expr{\app{\app{\var{append}}{\var{xs}}}{\brackets{\app{\app{\var{Cons}}{\var{x}}}{\brackets{\app{\app{\var{Cons}}{\var{y}}}{\var{Nil}}}}}}}$
\end{center}
Within this term, the list $xs$ has only to be traversed once, so a linear speedup has been obtained. This linear improvement will be made between each 
successive call of the na\"ive reverse function, thus giving an overall superlinear speedup and producing the accumulating reverse program as shown in 
Fig. \ref{example1}.
\end{example}
Although it appears that more and more efficiency improvements will be made as we move up our transformation hierarchy, in practice it is found that not many 
efficiency improvements are made beyond level 2, as it is unlikely that a programmer would write programs that are so inefficient that speedup factors greater than $O(n)$ are possible.
\section{Conclusion and Related Work}

We have presented a hierarchy of program transformers in which the transformer at each level of the hierarchy builds on top of those at lower levels.
We have proved that the transformers at each level in the hierarchy terminate, and have characterised the speedups that can be
obtained at each level. Previous works \cite{KOTT80,ANDERSEN92,AMTOFT93,ZHU94,SORENSEN94A} have noted that the unfold/fold transformation 
methodology is incomplete; some programs cannot be synthesised from each other. It is our hope that this work will help to overcome this restriction.

The seminal work corresponding to level 1 in our hierarchy is that of Turchin on supercompilation \cite{TURCHIN86}, although our level 1 transformer more
closely resembles positive supercompilation \cite{SORENSEN96}. 
There have been several previous attempts to move beyond level 1 in our transformation hierarchy, the first one by Turchin himself using {\em walk grammars} 
\cite{TURCHIN93}. In this approach, traces through residual graphs are represented by regular grammars that are subsequently analysed and simplified. 
This approach is also capable of achieving superlinear speedups, but no automatic procedure is defined for it; the outlined heuristics and strategies may not terminate.

A hierarchy of program specialisers is described in \cite{GLUCK98} that shows how programs can be metacoded and then manipulated through a 
{\em metasystem transition}, with a number of these metasystem transitions giving a metasytem hierarchy in which the original program may have several levels of
metacoding. In the work described here, a process tree can be considered to be the metacoding of a program. However, we do not have the difficulties associated 
with metasystem transitions and muli-level metacoding, as our process trees are residualised back to the object level.

Distillation \cite{HAMILTON07A,HAMILTON12A} is built on top of positive supercompilation, so corresponds to level 2 in our hierarchy, but does not go beyond this level.
Klyuchnikov and Romanenko \cite{KLYUCHNIKOV10B} construct a hierarchy of supercompilers  in which lower level supercompilers are used
to prove lemmas about term equivalences, and higher level supercompilers utilise these lemmas by rewriting
according to the term equivalences (similar to the ``second order replacement method'' defined by Kott \cite{KOTT85}).
Transformers in this hierarchy are capable of similar speedups to those in our hierarchy, but no automatic procedure is defined for it; the need 
to find and apply appropriate lemmas introduces infinite branching into the search space, and various heuristics have to be used to try to limit this search. 
Preliminary work on the hierarchy of transformers defined here was presented in \cite{HAMILTON12B}; this did not include analysis of the efficiency improvements
that can be made at each level in the hierarchy. The work described here is a lot further developed than that described in \cite{HAMILTON12B},
and we hope simpler and easier to follow.

Logic program transformation is closely related, and the equivalence of partial deduction and driving (as used in supercompilation)
has been argued by Gl{\"u}ck and S{\o}rensen \cite{GLUCK94}. Superlinear speedups can  be
achieved in logic program transformation by {\em goal replacement} \cite{PETTEROSSI04,ROYCHOUDHURI04}: replacing one logical 
clause with another to facilitate folding. Techniques similar to the notion of  ``higher level supercompilation'' \cite{KLYUCHNIKOV10B}
have been used to prove correctness of goal replacement, but have similar problems regarding the search for appropriate lemmas. 

\section*{Acknowledgements}

This work owes a lot to the input of Neil D. Jones, who provided many useful insights and ideas on the subject
matter presented here.
\bibliographystyle{splncs04}
\bibliography{mybib}

\end{document}